\documentclass[a4paper,12pt]{article}

\newtheorem{theorem}{Theorem}[section]
\newtheorem{lemma}[theorem]{Lemma}
\newtheorem{proposition}[theorem]{Proposition}

\newtheorem{problem}[theorem]{Problem}

\newenvironment{problem 1}[1][Problem 1]{\begin{trivlist}
\item[\hskip \labelsep {\bfseries #1}]}{\end{trivlist}}
\newenvironment{proof}[1][Proof]{\begin{trivlist}
\item[\hskip \labelsep {\bfseries #1}]}{\end{trivlist}}

\newenvironment{claim 1}[1][Claim 1]{\begin{trivlist}
\item[\hskip \labelsep {\bfseries #1}]}{\end{trivlist}}
\newenvironment{claim 2}[1][Claim 2]{\begin{trivlist}
\item[\hskip \labelsep {\bfseries #1}]}{\end{trivlist}}
\newenvironment{claim 3}[1][Claim 3]{\begin{trivlist}
\item[\hskip \labelsep {\bfseries #1}]}{\end{trivlist}}
\newenvironment{claim 4}[1][Claim 4]{\begin{trivlist}
\item[\hskip \labelsep {\bfseries #1}]}{\end{trivlist}}

\newcommand{\bal}{{\rm bal}} 

\newcommand{\qed}{\hfill$\Box$} 

\usepackage{amsfonts}
\usepackage{amsmath}
\usepackage{amssymb}  
\usepackage{amssymb,amsmath,graphics,color}
\usepackage{graphicx}
\usepackage{multicol}

\title{Parameterized TSP: Beating the Average\footnote{GG's research was partially supported by Royal Society Wolfson Research Merit Award}}
\author{Gregory Gutin\\Royal Holloway, University of London\\ e-mail: \texttt{gutin@cs.rhul.ac.uk}  \and Viresh Patel\\Queen Mary, University of London\\ e-mail: \texttt{viresh.patel@qmul.ac.uk} }
\begin{document}
\maketitle

\begin{abstract}
\noindent In the Travelling Salesman Problem (TSP), we are given a complete graph $K_n$ together with an integer weighting $w$ on the edges of $K_n$, and we are asked to find a Hamilton cycle of $K_n$ of minimum weight. Let $h(w)$ denote the average weight of a Hamilton cycle of $K_n$ for the weighting $w$. Vizing (1973) asked whether there is a polynomial-time algorithm which always finds a Hamilton cycle of weight at most $h(w)$. He answered this question in the affirmative and subsequently Rublineckii  (1973) and others described several other TSP heuristics satisfying this property. In this paper, we prove a considerable generalisation of Vizing's result: for each fixed $k$, we give an algorithm that decides whether, for any input edge weighting $w$ of $K_n$, there is a Hamilton cycle of $K_n$ of weight at most $h(w)-k$ (and constructs such a cycle if it exists). For $k$ fixed, the running time of the algorithm is polynomial in $n$, where the degree of the polynomial does not depend on $k$ (i.e., the generalised Vizing problem is fixed-parameter tractable with respect to the parameter $k$).
\end{abstract}

\section{Introduction}

The Travelling Salesman Problem (TSP) is one of the most well-known and widely studied combinatorial optimisation problems. In this problem, we are given an $n$-vertex complete graph $K_n$ with weights on its edges and we are required to find a Hamilton cycle in $K_n$ of minimum total weight. In its full generality,  TSP is not only NP-hard, but also NP-hard to approximate to within any constant factor. Therefore there has been much attention in developing approximation algorithms for restricted instances of TSP. 
In this paper, we consider general TSP, but rather than seeking a Hamilton cycle of minimum weight, we seek a Hamilton cycle that beats the average weight of all Hamilton cycles by some given value.

Let us fix some notation in order to state our result. As usual $V(G)$ and $E(G)$ denote the vertex and edge sets of a graph $G$. Let $w$ be an integer edge weighting of $K_n$, i.e.\ $w: E(K_n) \rightarrow \mathbb{Z}$, and let $G$ be a subgraph of $K_n$. We  write
\[
w(G):= \sum_{e \in E(G)} w(e)
\hspace{1 cm} \text{and} \hspace{1 cm}
w[G]:= \sum_{e \in E(G)} |w(e)| 
\]
and we define the {\it density} $d=d(w)$ of $w$ to be the average weight of an edge, i.e.\ $d:=w(K_n)/ \binom{n}{2}$.  
Note that $\mathbb{E}(w(\tilde{H})) = dn$, where $\tilde{H}$ is a uniformly random Hamilton cycle of $K_n$, and so there always exists a Hamilton cycle $H^*$ satisfying $w(H^*) \leq dn$. 

 Vizing \cite{Viz} asked whether there is a polynomial-time algorithm which, given an integer edge weighting $w$ of $K_n$, always finds a Hamilton cycle $H^*$ of $K_n$ satisfying $w(H) \leq dn$. He answered this question in the affirmative and subsequently Rublineckii \cite{Rub} described several other TSP heuristics satisfying this property. Such TSP heuristics including more recent ones are given in \cite{GutYeoZve}. A natural question extending Vizing's question is the following: for each fixed $k$ is there a polynomial-time algorithm which, given $w$, determines if there exists a Hamilton cycle $H^*$ satisfying $w(H^*) \leq dn - k$? We give an affirmative answer to this question.

\begin{theorem}
\label{th:main}
There exists an algorithm which, given $(n,w,k)$ as input, where $n,k \in \mathbb{N}$ and $w: E(K_n) \rightarrow \mathbb{Z}$, determines whether there exists a Hamilton cycle $H^*$ of $K_n$ satisfying 
\[
w(H^*) \leq dn - k
\]
(and outputs such a Hamilton cycle if it exists) in time $O(k^3)! + O(k^3n) + O(n^7) = f(k)n^{O(1)}$.
\end{theorem}
Note that our algorithm includes arithmetic operations which are assumed to take time $O(1)$ and so our running times here and throughout are stated in the \emph{strong} sense (see e.g.\ \cite{GroLovSch}). To obtain the running time in the \emph{weak} sense, one simply multiplies by $\log M$, where $M := \max_{e \in E(K_n)}|w(e)|$ for the input instance $(n,w,k)$.

Theorem~\ref{th:main} immediately implies that the following NP-hard problem (which is essentially TSP) is fixed-parameter tractable\footnote{For a recent introductions to parameterised algorithms and complexity, see monographs \cite{Cyg,DowFel}.} when parameterised by $k$.

\bigskip
\noindent
\textsc{Travelling Salesman Problem Below Average (${\rm TSP}_{\rm{BA}}$)} 

\medskip
\noindent
\begin{tabular}{p{1.7cm}p{11cm}}
\textit{Instance}\,:& $(n,w,k)$, where $n,k \in \mathbb{N}$ and $w: E(K_n) \rightarrow \mathbb{Z}$\\
\textit{Question}\,:&Is there a Hamilton cycle $H^*$ of $K_n$ satisfying $w(H^*) \leq dn - k$? 
\end{tabular}

\bigskip

Theorem~\ref{th:main} is proved by applying a combination of probabilistic, combinatorial, and algorithmic techniques, some of which are inspired by \cite{KuhOstPat}.
The key step to proving Theorem~\ref{th:main} is Theorem~\ref{th:main2} below, which characterises those weightings in which all Hamilton cycles have weight close to the average. We believe this result will have further applications. 
\begin{theorem}
\label{th:main2}
For any $n,k \in \mathbb{N}$ satisfying $n>5000(k+1)$, and $w:E(K_n) \rightarrow \mathbb{Z}$, in time $O(n^7)$ we can find either
\begin{itemize}
\item[{\rm (a)}] A Hamilton cycle $H^*$ of $K_n$ satisfying $w(H^*) < dn - k$, or
\item[{\rm (b)}] A weighting $w^*:E(K_n) \rightarrow \mathbb{Z}$ and $\alpha \in \mathbb{Z}$ satisfying $w^*(H) = w(H) + \alpha$ for all Hamilton cycles $H$ of $K_n$ and $w[K_n] \leq 4000kn$. 
\end{itemize}
\end{theorem}
Note that since $k$ is the parameter, it can be viewed as a constant, and so we may assume that $n\ge g(k)$ for any function $g$.

\medskip
\noindent
{\bf Related work}
The problem we consider in this paper falls into a class of problems introduced by Mahajan, Raman, and Sikdar \cite{MahRamSik}. The general framework is the following.
Consider a combinatorial optimisation problem in which one is
seeking a feasible solution of minimum (or maximum) \emph{value} and suppose further that one is always guaranteed to find a feasible solution whose \emph{value} is at most (or at least) some non-trivial (often tight) bound $b$ (e.g.\ in our case, for any instance of TSP, one can always find a Hamilton cycle of weight at most $dn$ where $d$ is the average weight of an edge of $K_n$). One can then consider the problem, parameterised by $k$, of determining whether there exists a feasible solution of value at most $b - k$ (or at least $b+k$). 
A variety of techniques combining tools from linear algebra, the probabilistic method, Harmonic analysis, combinatorics and graph theory have been applied to such problems; see \cite{GutYeo12} for a survey. Here, we mention progress on only a small sample of such problems.

For the Maximum $r$-Satisfiability Problem, given a multiset of $m$ clauses of size $r$, a straigtforward probabilistic argument shows that there exists a truth assignment satisfying at least $(2^r - 1)m/ 2^{r}$ clauses and this is tight. Alon et al.\ \cite{AGKSY} showed that one can decide in time $O(m) +  2^{O(k^2)}$ if there is a truth assignment satisfying at least $((2^r -1)m + k) / 2^r$ clauses,  where they used a combination of probabilistic, combinatorial and Harmonic analysis tools. 

For the Max-Cut problem, the Edwards-Erd{\"o}s bound \cite{Ed1, Ed2} states that every connected graph on $n$ vertices and $m$ edges has a cut of size at least $\frac{m}{2} + \frac{n-1}{4}$ and this is tight. Crowston, Jones, and Mnich \cite{CroJonMni} showed that it is fixed-parameter tractable to decide whether a given graph on $n$ vertices and $m$ edges has a cut of size at least $\frac{m}{2} + \frac{n-1}{4} +k$. This was later extended by Mnich et al. \cite{MniPhiSauSuc} to so-called \emph{$\lambda$-extendible properties}; as special cases of their result, they could extend the Max-Cut result above to the Max $q$-Colourable Subgraph problem and the Oriented Max Acyclic Digraph problem.

\bigskip
\noindent
{\bf Organisation} 
The rest of the paper is organised as follows. In the next section we set out the notation we use throughout. Section~\ref{se:overview} gives a brief discussion of some the ideas that are used to prove Theorem~\ref{th:main} and Theorem~\ref{th:main2}.
In Section~\ref{se:derandomisation}, we show how standard derandomisation techniques can be applied to the Travelling Salesman Problem in preparation for Sections~\ref{sec:stab} and~\ref{se:algorithms}. Section~\ref{sec:stab} is dedicated to the proof of Theorem~\ref{th:main2} and this is used in Section~\ref{se:algorithms} to prove Theorem~\ref{th:main}.


\section{Notation and Terminology}
\label{se:notation}

In this section, for convenience, we collect some notation and terminology (mostly standard) that we shall use throughout.

Let $G$ be a graph. The vertex set and edge set of $G$ are denoted by $V(G)$ and $E(G)$ respectively and we write $e(G)$ for the number of edges in $G$. A graph $F$ is a subgraph of $G$ written $F \subseteq G$ if $V(F) \subseteq V(G)$ and $E(F) \subseteq E(G)$. We say $F$ is a spanning subgraph of $G$ if $V(F) =V(G)$ and $E(F) \subseteq E(G)$. 

For $X \subseteq V(G)$, we write $X^{(2)}$ for the set of all edges $ab$ such that $a,b \in X$ and $a \not= b$.
 We write $G[X]$ for the graph induced by $G$ on $X$ and $G - X$ for the graph obtained from $G$ by deleting all vertices in $X$ i.e.\ $G-X := G[V(G) \setminus X]$. For $S \subseteq E(G)$, $G-S$ is the graph obtained from $G$ by deleting all the edges in $S$, i.e.\ $G-S := (V(G), E(G) \setminus S)$.
If $S \subset V(G)^{(2)}$ then we write $G \cup S := (V(G), E(G) \cup S)$ (and we write $G \cup e$ rather than $G \cup  \{e\}$ if $S=\{e\}$). 
 For disjoint subsets $A,B$  of $V(G)$, we write $G[A,B]$ for the graph with vertex set $A \cup B$ and edge set $\{e=ab \in E(G) \mid a \in A, b \in B\}$. For a vertex $v \in V(G)$, $N_G(v)$ denotes the set of neighbours of $v$ in $G$ and $d_v(G) := |N_G(v)|$ denotes the degree of $v$. The maximum and minimum degree of $G$ is denoted by $\delta(G)$ and $\Delta(G)$ respectively. 

A path $P=v_1v_2 \cdots v_k$ is the graph with vertices $v_1 \ldots, v_k$ and edges $v_iv_{i+1}$ for $i = 1, \ldots, k-2$.
For a path $P$, we sometimes write $v_1Pv_k$ for the same path to indicate that $v_1$ and $v_k$ are its end-vertices and we say $v_2, \ldots, v_{k-1}$ are the internal vertices of $P$. The notation extends in the natural way for concatenated paths so that if $v_1Pv_k$ and $w_1Qw_{\ell}$ are paths and $x_1, \ldots, x_t$ are vertices, then $v_1Pv_kx_1 \cdots x_tw_1Qw_{\ell}$ is the path
$v_1 \cdots v_kx_1 \cdots x_tw_1 \cdots w_{\ell}$. 
A cycle $C = v_1\cdots v_k v_1$ is the graph with vertices $v_1 \ldots, v_k$ and edges $v_1v_k$ and $v_iv_{i+1}$ for $i = 1, \ldots, k-2$. We call it a $k$-cycle if it has $k$ vertices. A cycle that is a spanning subgraph of a graph $G$ is called a Hamilton cycle of $G$. As before, we can write a cycle as a concatenation of paths.
 A matching of $G$ is a subgraph of $G$ of maximum degree $1$; a perfect matching of $G$ is a spanning matching of $G$. The complete graph on $n$ vertices is denoted by $K_n$.

Repeating notation from the introduction, recall that an instance of ${\rm TSP}_{\rm{BA}}$ consists of a triple $(n,w,k)$, where $n,k \in \mathbb{N}$ and $w: E(K_n) \rightarrow \mathbb{Z}$. We sometimes drop the parameter $k$ (when it is not relevant) and refer to instances $(n,w)$. 
For a subgraph $G$ of $K_n$, we write
\[
w(G):= \sum_{e \in E(G)} w(e)
\hspace{1 cm} \text{and} \hspace{1 cm}
w[G]:= \sum_{e \in E(G)} |w(e)| 
\]
and we define the {\it density} $d=d(n,w)$ of $(n,w)$ to be the average weight of an edge, i.e.\ $d:=w(K_n)/ \binom{n}{2}$.



\section{Overview}
\label{se:overview}

We remark at the outset that the discussion in this section is not required to understand the sections that follow; some definitions will be repeated later. 

\medskip
\noindent
{\bf Structural result} The key step for the algorithm of  Theorem~\ref{th:main} is the structural result, Theorem~\ref{th:main2}.   In order to explain the idea behind its proof, let us recast Theorem~\ref{th:main2} in the language of norms.


We say two instances $(n,w)$ and $(n,w')$ of ${\rm TSP}_{\rm{BA}}$ are equivalent, written $(n,w) \sim (n,w')$, if there exists some $\alpha \in \mathbb{Z}$ such that $w'(H) = w(H) + \alpha$ for every Hamilton cycle $H$ of $K_n$. 
We define 
\[
\|(n,w)\|_{1 / \sim}:= \min \{w'[K_n]: w' \sim w \}.
\]
For an instance $(n,w)$ of density $d$, if $k^*$ is such that $dn - k^*$ is the weight of a minimum weight Hamilton cycle of $K_n$, we define $\|(n,w)\|_{\rm HC}:=k^*$. 
Then the main substance of Theorem~\ref{th:main2} is that the following inequality holds: 
\begin{equation}
\label{eq:structure}
\|(n,w)\|_{1 / \sim} \leq 4000n\|(n,w)\|_{\rm HC}.
\end{equation}
This is proved by considering a third parameter $\|(n,w) \|_{4-{\rm cyc}}$ that is easily computed by examining the $4$-cycles of $K_n$. This parameter is introduced in Section~\ref{sec:stab}, and in the same section we prove the two inequalities
\begin{align}
\|(n,w)\|_{1 / \sim} &\leq \frac{4000}{n^2}\|(n,w) \|_{4-{\rm cyc}} \label{eq:s1} \\
\frac{1}{n^3}\|(n,w) \|_{4-{\rm cyc}} &\leq \|(n,w)\|_{\rm HC},
\label{eq:s2}
\end{align} 
which together prove (\ref{eq:structure}).

We make some further remarks. Note that for fixed $n$, the set of instances $(n,w)$ with the obvious notions of addition and scalar multiplication is the vector space $\mathbb{R}^{E(K_n)}$. One can show that $\sim$ is an equivalence relation and the equivalence classes are translates of $Z$, the set of instances equivalent to the all-zero weighting, which turns out to be an $(n-1)$-dimensional subspace of $\mathbb{R}^{E(K_n)}$.
Furthermore, each of $\|\cdot\|_{1 / \sim}$, $\|\cdot \|_{4-{\rm cyc}}$, and  $\|\cdot\|_{\rm HC}$ are pseudo-norms on $\mathbb{R}^{E(K_n)}$ and norms on the quotient space $\mathbb{R}^{E(K_n)} / Z$. In particular, $Z$ is precisely the set of instances $(n,w)$ in which all Hamilton cycles have the same weight $dn$, where $d$ is the density of $(n,w)$. 

\medskip
\noindent
{\bf Algorithmic result} Once we have established the structural result Theorem~\ref{th:main2}, we can construct the algorithm of Theorem~\ref{th:main}. Given $(n,w,k)$, if Theorem~\ref{th:main2} does not already give us the desired Hamilton cycle of weight at most $dn-k$, then we can find an equivalent instance $(n,w',k)$ with $w'[K_n] \leq 4000kn$. If we can find a large matching in the graph of edges assigned a negative weight by $w'$ then we can extend it (in a random way) to a desired Hamilton cycle of low weight. If such a matching does not exist, then one can easily conclude that all edges asigned a negative weight by $w'$ are incident with only a small set (of size depending on $k$ but independent of $n$) of vertices. It turns out that, with this additional structure, one can in fact find the minimum weight Hamilton cycle for $(n,w')$ (and hence the minimum weight Hamilton cycle for $(n,w)$) in FPT-time i.e.\ time $f(k)n^{O(1)}$, where $f$ is a function of $k$ only.  

\section{Derandomisation}
\label{se:derandomisation}

In this section we present some standard derandomisation arguments that we shall require later. First some notation.

We denote by $\mathcal{H}_n$ the set of all Hamilton cycles of the complete graph $K_n$.
For any graph (or set of edges)  $G$, let $\mathcal{H}_n^G:= \{H \in \mathcal{H}_n \mid G \subseteq H\}$. 
In general we shall denote by $\tilde{H}$ a uniformly random element of $\mathcal{H}_n$, and by $\tilde{H}^G$ a uniformly random element of $\mathcal{H}_n^G$.

 We say a graph $G \subseteq K_n$ is a {\em partial Hamilton cycle} of $K_n$ if $G$ is a spanning subgraph of some $H \in \mathcal{H}_n$; thus $G$ is either a Hamilton cycle or the union of vertex disjoint paths (where we allow a path to be a singleton vertex). A path consisting of a single vertex is called a {\em trivial path}, and a path on several vertices is called a {\em non-trivial path}. We shall use the following simple fact several times: if $G$ is a partial Hamilton cycle with $r$ non-trivial paths and $s$ trivial paths then
 \begin{equation}
 \label{eq:Hn}
 |\mathcal{H}_n^G| = 2^{r-1}(r+s-1)!. 
 \end{equation}

For $G$ a partial Hamilton cycle of $K_n$, we denote by $J(G)$ the set of edges in $K_n$ which join two paths of $G$ together into a single path. If $G$ is a Hamilton path, $J(G)$ is defined to be the unique edge between the two ends of the path.

\begin{lemma}
\label{le:derand1}
Suppose we have a function $X: \mathcal{H}_n \rightarrow \mathbb{Q}$ and for every partial Hamilton cycle $G$ of $K_n$, assume we can compute $\mathbb{E}(X(\tilde{H}^G))$ in time $f(n)$. Then for any given partial Hamilton cycle $G^*$, we can find in time $O(n^3f(n))$ a Hamilton cycle $H^* \in \mathcal{H}_n^{G^*}$ such that $X(H^*) \leq \mathbb{E}(X(\tilde{H}^{G^*}))$.
\end{lemma}
\begin{proof}
From the law of total expectation, we have
\begin{align*}
\mathbb{E}(X(\tilde{H}^G)) 
&= \sum_{e \in J(G)} \mathbb{P}(e \in \tilde{H}^G)
\;\mathbb{E}(X(\tilde{H}^{G \cup e})) 
\end{align*}
and so we know there exists some $e \in J(G)$ such that $\mathbb{E}(X(\tilde{H}^{G \cup e})) \leq \mathbb{E}(X(\tilde{H}^{G}))$.

We construct $H^*$ by adding edges one at a time to $G^*$ as follows. Assume $G^*$ has $q$ edges for some $q \geq 0$ and set $G_q:=G^*$. Assume we have constructed a partial Hamilton cycle $G_{q'} \supseteq G_q$ with $q' \geq q$ edges satisfying
$\mathbb{E}(X(\tilde{H}^{G_{q'}})) \leq \mathbb{E}(X(\tilde{H}^{G_q}))$. For each $e \in J(G_{q'})$ we compute $\mathbb{E}(X(\tilde{H}^{G_{q'} \cup e}))$ and determine an edge $e^*$ for which 
$\mathbb{E}(X(\tilde{H}^{G_{q'} \cup e^*}))
\leq 
\mathbb{E}(X(\tilde{H}^{G_{q'}}))$. This can be done in time $O(f(n)n^2)$. We set $G_{q'+1}:= G_{q'} \cup e^*$. After at most $n$ iterations of this process, we obtain a Hamilton cycle $H^*$ satisfying the desired condition. The running time is therefore bounded by $O(f(n)n^3)$.
\qed
\end{proof}

\begin{lemma}
\label{le:derand2}
Given any instance $(n,w)$ and any partial Hamilton cycle $G$ of $K_n$, we can find in time $O(n^5)$ a Hamilton cycle $H^* \in \mathcal{H}_n^G$ such that $w(H^*) \leq \mathbb{E}(w(\tilde{H}^G))$.
\end{lemma}
\begin{proof}
Apply the previous lemma. We can compute $\mathbb{E}(w(\tilde{H}^G))$ in time $O(n^2)$. Indeed, note that
\[
\mathbb{E}(w(\tilde{H}^G)) = w(G) + \sum_{e \in J(G)}\mathbb{P}(e \in \tilde{H}^G)w(e)
\]
and using (\ref{eq:Hn}),
\begin{equation*}
\mathbb{P}(e \in \tilde{H}^G) 
= \frac{|\mathcal{H}_n^{G \cup e}|}{|\mathcal{H}_n^G|}
= \frac{2^{r'-1}(r'+s'-1)!}{2^{r-1}(r+s-1)!},
\end{equation*}	
where $r,s$ are the numbers of non-trivial and trivial paths in $G$ and $r',s'$ are the numbers of non-trival and trivial paths in $G \cup e$. In fact it is not hard to see that if $e \in J(G)$, then $r+s = r'+s'+1$, and so
\begin{equation}
\label{eq:prob}
\mathbb{P}(e \in \tilde{H}^G) = \frac{2^{r'-r}}{r+s-1}.
\end{equation}
Thus since $|J(G)| = O(n^2)$, we can compute $\mathbb{E}(w(\tilde{H}^G))$ in time $O(n^2)$ as required.
\qed
\end{proof}

\section{The structural result}
\label{sec:stab}
Our aim in this section is to prove Theorem~\ref{th:main2}. 
In this section, an instance refers to a pair $(n,w)$ where $w:E(K_n) \rightarrow \mathbb{Z}$. 
Let us denote the set of $4$-cycles of $K_n$ by $\mathcal{C}_n$.
Given an instance $(n,w)$ and a $4$-cycle $C = v_1v_2v_3v_4v_1 \in \mathcal{C}_n$, we define the \emph{balance} of $C$ (with respect to $(n,w)$) to be
\[
\bal_{(n,w)}(C) = \bal(C):= |w(v_1v_2) + w(v_3v_4) - w(v_1v_3) - w(v_2v_4)|.
\]
We say that $C$ is {\em balanced} if $\bal(C)=0$; otherwise we say $C$ is {\em unbalanced}. 
For a set $\mathcal{A} \subseteq \mathcal{C}_n$, we define
\[
\bal(\mathcal{A}) := \sum_{C \in \mathcal{A}} \bal(C),
\]
and we set 
\[
\|(n,w)\|_{4-\rm{cyc}} := \bal(\mathcal{C}_n) = \sum_{C \in \mathcal{C}_n} \bal(C).
\]
Our first lemma gives a polynomial-time witness to the inequality (\ref{eq:s2}).

\begin{lemma}
\label{le:ave}
For a given instance $(n,w)$ and $k \in \mathbb{Q}$, suppose $\|(n,w)\|_{4-\rm{cyc}} \geq kn^3$. Then there exists a Hamilton cycle $H$ satisfying $w(H) < dn -k$. Furthermore, we can find such a Hamilton cycle in time $O(n^7)$.
\end{lemma}

Before we can prove this lemma, we require several preliminary results.
Suppose $H = v_1v_2 \cdots v_nv_1$ is a Hamilton cycle of $K_n$. Consider two edges $e_1=v_av_b$ and $e_2=v_xv_y$ of $K_n$, where without loss of generality $a<b$, $x<y$, and $a<x$. We say $e_1$ and $e_2$ are {\em crossing (relative to $H$)} if $a<x<b<y$; otherwise we say $e_1$ and $e_2$ are non-crossing. This is just the natural notion of crossing in a planar drawing of $H$, $e_1$, and $e_2$.

\begin{proposition}
\label{pr:cross}
Suppose $H=v_1v_2 \cdots v_nv_1$ is a Hamilton cycle of $K_n$ and suppose we have a function $t: E(K_n) \rightarrow \{0,1,2,\ldots\}$. 
Then we can find $T^* \subseteq E(K_n)$ such that every pair of edges in $T^*$ is non-crossing and $t(T^*)\ge t(K_n)/(2n)$.
Furthermore, we can find $T^*$ in time $O(n^2)$.
\end{proposition}
\begin{proof}
Let $Q= \{3,4, \ldots, 2n \}$, and for each $q \in Q$, let $E_q = \{v_iv_j \in E(K_n) \mid  i+j=q\}$. 
Observe that the sets $(E_q)_{q \in Q}$ partition $E(K_n)$ and that for each 
fixed $q \in Q$ each pair of edges in $E_q$ is non-crossing. 
Thus since 
\[
t(K_n) = \sum_{q \in Q} t(E_q),
\]
there is some $q^* \in Q$ for which $t(E_{q^*}) \geq t(K_n)/|Q| \geq t(K_n)/2n$. Set $T^* =  E_{q^*}$.

We can determine $E_q$ and $t(E_q)$ in time $O(n^2)$ and so we can determine $T^*$ is time $O(n^2)$.
\qed
\end{proof}

We introduce some more terminology. Let $(n,w)$ be an instance and let $H$ be a Hamilton cycle of $K_n$. We say a $4$-cycle $C = v_1v_2v_3v_4v_1$ of $K_n$ is {\em embedded} in $H$ if either $v_1v_2, v_3v_4 \in E(H)$ or $v_2v_3, v_4v_1 \in E(H)$; note that if three edges of $C$ are in $H$, then $C$ is necessarily embedded in $H$. 
We say $C$ is {\em correctly embedded} in $H$ if $E(H) \triangle E(C)$ forms a Hamilton cycle; note that this happens if and only if $C$ is embedded in $H$, exactly two edges of $C$ belong to $H$, and the other two edges of $C$ are crossing relative to $H$. If $C$ is an unbalanced $4$-cycle and is embedded in $H$, we say it is {\em heavily embedded in $H$} if the edges of the heavier perfect matching of $C$ appear in $H$. We write $\mathcal{C}^*_n(H)$ for the set of $4$-cycles heavily and correctly embedded in $H$.
We define
\[
q(H):= \sum_{C \in \mathcal{C}^*_n(H)} \bal(C). 
\] 
The next lemma shows how a Hamilton cycle $H$ can be locally improved using $4$-cycles heavily and correctly embedded in it.

\begin{lemma}
\label{le:qh}
Let $(n,w)$ be an instance.
Given a Hamilton cycle $H$ of $K_n$, there exists  another Hamilton cycle $H'$ such that $w(H') \leq w(H) - (q(H)/2n)$, and we can find in it time $O(n^4)$. 
\end{lemma} 
\begin{proof}
Let $e_1, \ldots, e_n$ be the edges of $H$ in order. Suppose $C_1$ and $C_2$ are $4$-cycles that are correctly embedded in $H$, where $e_a,e_b$ with $a<b$ are the edges of $C_1$ in $H$ and $e_x,e_y$ with $x<y$ are the edges of $C_2$ in $H$, and without loss of generality, assume $a<x$. 
 We say that $C_1$ and $C_2$ are {\em crossing (relative to $H$)} if $a<x<b<y$.
 If $C_1$ and $C_2$ are not crossing relative to $H$, we find that $E(H) \triangle E(C_1) \triangle E(C_2)$ forms a Hamilton cycle. More generally, one can check that
if $C_1, \ldots, C_r$ are all $4$-cycles correctly embedded in $H$ and no pair of these $4$-cycles are crossing relative to $H$, then $E(H) \triangle E(C_1) \triangle \cdots \triangle E(C_r)$ forms a Hamilton cycle.

Consider an auxiliary complete graph $K^{\circ}_n$ with vertex set $e_1, \ldots, e_n$ and an auxiliary Hamilton cycle $H^{\circ}=e_1e_2 \cdots e_ne_1$. For each edge $e_xe_y$ of $E(K^{\circ}_n) \setminus E(H^{\circ})$ let $C_{xy}$ be the unique $4$-cycle of $K_n$ that is correctly embedded in $H$ and satisfies $E(H) \cap E(C) = \{e_x,e_y\}$. This correspondence is clearly bijective and crossing edges of $K^{\circ}_n$ (relative to $H^{\circ}$) correspond to crossing $4$-cycles of $K_n$ (relative to $H$). 
Define $t:E(K^{\circ}_n) \rightarrow \{0,1,\ldots\}$ by
\begin{align*}
t(e_xe_y) = 
\begin{cases}
\bal(C_{xy}) &\text{if } C_{xy} \text{ is heavily embedded in } H; \\
0 &\text{otherwise},
\end{cases}
\end{align*}
and note $t(K^{\circ}_n)=q(H)$.
Applying Proposition~\ref{pr:cross}, we can find a set of edges $T^* \subseteq E(K_n^{\circ})$ which are pairwise non-crossing relative to $H^{\circ}$ and satisfying $t(T^*) > q(H)/2n$. 
By removing from $T^*$ any edges assigned weight $0$ by $t$, we may further assume that every edge in $T^*$ is assigned a positive weight by $t$.
The edges of $T^*$ then correspond to $4$-cycles $C_1, \ldots, C_r$ of $K_n$ that are heavily and correctly embedded in $H$ such that no pair of these $4$-cycles are crossing (relative to $H$) and where $\bal(C_1) + \cdots + \bal(C_r) \geq q(H)/2n$. The time needed to find $C_1, \ldots, C_r$ is $O(n^2)$.

 We construct a sequence of Hamilton cycles $H_0, \ldots, H_r$, where $H_0=H$ and $E(H_i) = E(H_{i-1}) \triangle E(C_i)$ (which takes time $O(r)=O(n)$). 
 We have $w(H_i) = w(H_{i-1})- \bal(C_i)$ since $C_i$ is heavily embedded in $H_{i-1}$. Hence setting $H'=H_r$, we have \[
 w(H') = w(H) - \bal(C_1) - \cdots - \bal(C_r) \leq w(H) - (q(H)/2n),
\]  
 and the time required to construct $H'$ is $O(n^2)$. 
\qed
\end{proof}

The previous lemma motivates our interest in the function $q: \mathcal{H}_n \rightarrow \mathbb{Z}$. We shall later apply Lemma~\ref{le:derand1} to $q$ and for this we require the following straightforward proposition. We spell out the details for completeness.

\begin{proposition}
\label{pr:Eqh}
Let $(n,w)$ be an instance and let $G$ be a partial Hamilton cycle of $K_n$. Then we can compute $\mathbb{E}(q(\tilde{H}^G_n))$ in time $O(n^4)$.
\end{proposition}
\begin{proof}

Note that
\[
\mathbb{E}(q(\tilde{H}^G_n)) = \sum_{C \in \mathcal{C}_n}
\mathbb{P}( C \in \mathcal{C}^*_n(\tilde{H}^G) )\bal(C),
\]
and so it is sufficient to show how to compute $\mathbb{P}( C \in \mathcal{C}^*_n(\tilde{H}^G) )$ in time $O(1)$ for each $C \in \mathcal{C}_n$. This is intuitively clear, but slightly tedious to explain.

Clearly if $C$ is balanced then $\mathbb{P}( C \in \mathcal{C}^*_n(\tilde{H}^G) ) = 0$. So assume
$C = v_1v_2v_3v_4v_1 \in \mathcal{C}_n$ is unbalanced and that $e_1=v_1v_2$ and $e_2=v_3v_4$ are the edges of the heavier perfect matching of $C$. Then (using (\ref{eq:Hn})) the probability $p(C,G)$ that $C$ is heavily embedded in $\tilde{H^G}$ is 
\[
p(C,G) = \mathbb{P}(e_1,e_2 \in \tilde{H}^G) = \frac{|\mathcal{H}_n^{G \cup \{e_1,e_2\}}|}{|\mathcal{H}_n^G|}=
\frac{2^{r'-1}(r'+s'-1)!}{2^{r-1}(r+s-1)!},
\]
where $r$ and $s$ are the numbers of non-trivial and trivial paths in $G$ and $r'$ and $s'$ are the numbers of  non-trivial and trivial paths in $G \cup \{e_1,e_2\}$.

If $e_1$ and $e_2$ appear on different non-trivial paths of $G':=G \cup \{e_1,e_2 \}$, say $a_1P_1b_1$ and $a_2P_2b_2$, then $\mathbb{P}( C \in \mathcal{C}^*_n(\tilde{H}^G) )= \frac{1}{2}p(C,G)$. To see this, we note that a Hamilton cycle $H \in \mathcal{H}_n^{G'}$  can take one of two forms depending on relative orientations of $P_1$ and $P_2$ on $H$. 
In one of these forms $C$ is correctly embedded in $H$ and in the other it is not. Thus $C$ is correctly embedded in exactly half the Hamilton cycles of $\mathcal{H}_n^{G'}$.

If $e_1$ and $e_2$ appear on the same non-trivial path of $G'$ then either $\mathbb{P}( C \in \mathcal{C}^*_n(\tilde{H}^G) )=p(C,G)$ or $\mathbb{P}( C \in \mathcal{C}^*_n(\tilde{H}^G) )=0$ depending on the order of the vertices $v_1, \ldots, v_4$ on the path. Hence we have shown how to compute $\mathbb{P}( C \in \mathcal{C}^*_n(\tilde{H}^G) )$ in time $O(1)$.
\qed
\end{proof}

We are now ready to prove Lemma~\ref{le:ave}.
\begin{proof} (of Lemma~\ref{le:ave})
Let $\tilde{H}$ be a uniformly random Hamilton cycle of $K_n$ and suppose $C=v_1v_2v_3v_4v_1$ is an unbalanced $4$-cycle, with $e_1=v_1v_2$ and $e_2=v_3v_4$ the edges of its heavier perfect matching. Then the probability that $C$ is heavily embedded in $\tilde{H}$ is given by
\begin{equation*}
\mathbb{P}(e_1,e_2 \in \tilde{H}) = \frac{|\mathcal{H}_n^{\{e_1,e_2\}}|}{|\mathcal{H}_n|} = \frac{2(n-3)!}{(n-1)!/2} = \frac{4}{(n-1)(n-2)} \geq \frac{4}{n^2}.
\end{equation*}
Now (as in the proof of Proposition~\ref{pr:Eqh})
it is not too hard to see that the probability that $C$ is heavily and correctly embedded in $\tilde{H}$ is $\frac{1}{2}\mathbb{P}(e_1,e_2 \in \tilde{H}) \geq \frac{2}{n^2}$. 
Thus we have that
\begin{align*}
\mathbb{E}(q(\tilde{H})) = \sum_{C \in \mathcal{C}_n}
\mathbb{P}( C \in \mathcal{C}^*_n(\tilde{H}) )\bal(C) \geq \frac{2}{n^2} \| (n,w) \|_{4-{\rm cyc}} > 2kn.
\end{align*}
and so
\begin{align*}
\mathbb{E}(w(\tilde{H}) - (q(\tilde{H})/2n)) < dn - k. 
\end{align*}
Thus there exists a Hamilton cycle $H^*$ satisfying $w(H^*) - (q(H^*)/2n) < dn -k$. Furthermore we can find this Hamilton cycle in time $O(n^7)$. Indeed, by Lemma~\ref{le:derand1}, it is sufficient to check that we can compute $\mathbb{E}(w(\tilde{H}^G) - (q(\tilde{H}^G)/2n)) = \mathbb{E}(w(\tilde{H}^G)) - \mathbb{E}(q(\tilde{H}^G))/2n$ in time $O(n^4)$ for every partial Hamilton cycle $G$. This holds since we can compute $\mathbb{E}(w(\tilde{H}^G))$ in $O(n^2)$ time by the proof of Lemma~\ref{le:derand2} and we can compute $\mathbb{E}(q(\tilde{H}^G))$ in time $O(n^4)$ by Proposition~\ref{pr:Eqh}.

Finally, we apply Lemma~\ref{le:qh} to $H^*$ to obtain a Hamilton cycle $H$ satisfying $w(H) \leq w(H^*) - (q(H^*)/2n) < dn - k$ as required. This takes time $O(n^4)$, so the total running time is $O(n^7)$.
\qed 
\end{proof}

Our next task is to prove (\ref{eq:s1}) and give a polynomial-time witness for the inequality.
Recall that two instances $(n,w)$ and $(n,w')$ are equivalent if there exists some $\alpha \in \mathbb{Z}$ such that $w'(H) = w(H) + \alpha$ for all Hamilton cycles of $K_n$.



\begin{lemma}
\label{le:norm}
Suppose $(n,w)$ is an instance and $0 \leq k \in \mathbb{Q}$ with $n> 5000(k+1)$. If $\|(n,w)\|_{4-{\rm cyc}} \leq kn^3$ 
then there exists an equivalent instance $(n,w^*)$ satisfying $w^*[K_n] \leq 4000kn$. Moreover, we can determine $(n,w^*)$ in time $O(kn^5)$.
\end{lemma}

After proving two preliminary results, we will prove Lemma~\ref{le:reduce}, which states that if $\|(n,w)\|_{4-{\rm cyc}}$ is small then there is an equivalent instance $(n,w^*)$ in which most edges of $K_n$ are assigned weight $0$. From this, we shall deduce Lemma~\ref{le:norm}.


\begin{lemma}
\label{le:S}
Suppose $(n,w)$ is an instance and $k \in \mathbb{Q}$ with $n > 50k$ such that $\|(n,w)\|_{4-{\rm cyc}} \leq kn^3$.
Then there exists a set $S$ of at most $16kn$ edges of $K_n$ such that every $4$-cycle of $K_n - S$ is balanced. Moreover, we can find $S$ in time $O(kn^5)$.
\end{lemma}
\begin{proof}
Note first that if $\|(n,w)\|_{4-{\rm cyc}} \leq kn^3$ then there are at most $kn^3$ unbalanced $4$-cycles.

 Suppose $C = v_1v_2v_3v_4v_1$ is an unbalanced $4$-cycle of $K_n$ and let $ab$ be any edge of $K_n$ that is vertex disjoint from $C$. Then one of the following $4$-cycles is unbalanced:
$av_1v_2ba$, $bv_2v_3ab$, $av_3v_4ba$, $bv_4v_1ab$. Indeed, if not then we have 
\begin{align*}
w(ab) 	&= w(av_1) + w(bv_2) - w(v_1v_2) \\
		&= w(bv_2) + w(av_3) - w(v_2v_3) \\
		&= w(av_3) + w(bv_4) - w(v_3v_4) \\
		&= w(bv_4) + w(av_1) - w(v_4v_1). 
\end{align*}
But then $w(v_1v_2) + w(v_3v_4) = w(v_1v_3) + w(v_2v_4)$ and so $C$ is a balanced $4$-cycle, a contradiction. Hence we see that for every $ab \in E(K_n - V(C))$, there exists an unbalanced $4$-cycle $C'$ which contains $ab$ and shares an edge with $C$.

Now consider the greedy process where we start with $G_0:=K_n$ and we iteratively remove the edges of an arbitrary unbalanced $4$-cycle $C_r$ from $G_r$ to obtain $G_{r+1}$; thus $G_r$ has $\binom{n}{2}- 4r$ edges. This process must eventually stop, say with the graph $G_{r^*}$ that contains no unbalanced $4$-cycles. Let $q_r$ be the number of unbalanced $4$-cycles in $G_r$. Then we have for all $r < r^*$ that
\[
q_{r+1} \leq q_r - e(G_{r} - V(C_{r})) \leq q_r - \left( \binom{n-4}{2} - 4r \right);
\]
the first inequality holds because, as observed above, for every $ab \in E(G_r - V(C_{r}))$, there exists an unbalanced $4$-cycle of $G_r$ that contains $ab$ and shares an edge with $C_{r}$ (and this $4$-cycle is therefore absent from $G_{r+1}$).
We deduce that $q_r \leq q_0 - r\binom{n-4}{2} + 4\binom{r}{2}$. Since $q_0 \leq kn^3$, if $r^* \geq 4kn$, then \begin{align*}
q_{r^*} \leq kn^3 - 4kn \binom{n-4}{2} + 4\binom{4kn}{2} < 0,
\end{align*}
where the last inequality holds for our choice of $n$ large enough. Thus $r^* \leq 4kn$, and setting $S$ to be the set of at most $16kn$ edges removed from $K_n$ to obtain $G_{r^*}$, we see that all $4$-cycles of $K_n - S$ are balanced.

The proof is constructive and we can find $S$ in time $O(kn^5)$: it takes time $O(n^4)$ to search for and remove an unbalanced $4$-cycle and we do this at most $4kn$ times, so the total running time is $O(kn^5)$.
\qed
\end{proof}

The following proposition is straightforward, but we spell out the details for completeness.
\begin{proposition}
\label{pr:component}
Suppose $G$ is a graph on $n$ vertices and at least $\binom{n}{2}-t$ edges for some $0 \leq t \leq \binom{n}{2}$. Then $G$ has a connected component with at least $\binom{n}{2} - 2t$ edges.
\end{proposition}
\begin{proof}
Let $A \subseteq V(G)$ be a connected component of $G$ with the maximum number of vertices. Then it must be the case that $|A| \geq n - 1 - (2t/n)$; if not then $\Delta(G) < n - 1 - (2t/n)$, which implies $e(G) \leq n\Delta(G)/2 < \binom{n}{2} - t$, a contradiction. 
Hence
\[
e(G[A]) \geq e(G) - \binom{n - |A|}{2} \geq \binom{n}{2} - t - \binom{1 + (2t/n)}{2} \geq 
 \binom{n}{2} - 2t,
\]
where the last inequality follows using that $\binom{n}{2} \geq t$.
\qed
\end{proof}

Next, we describe simple linear operations one can apply to an instance to obtain an equivalent instance.
For a vertex  $v$ of $K_n$ we write $I_v:E(K_n) \rightarrow \{0,1 \}$ for the edge weighting of $K_n$ where $I_v(e)=1$ if $v$ is an end-vertex of $e$, and $I_v(e)=0$ otherwise. Observe that if $(n,w)$ is an instance and 
$w' = w + \lambda I_v$ for some $\lambda \in \mathbb{Z}$, then $w'(H) = w(H) + 2 \lambda$ for all Hamilton cycles $H$ of $K_n$ and both $w$ and $w'$ have the same balanced $4$-cycles. 
 More generally, given an integer $\lambda_v$ for each $v \in V(K_n)$, if
\begin{equation}
\label{eq:1}
 w' = w + \sum_{v \in V(K_n)} \lambda_v I_v
\end{equation}
then $w'(H) = w(H) + 2\sum_v \lambda_v$ for all Hamilton cycles $H$ of $K_n$
and so $(n,w)$ and $(n,w')$ are equivalent instances. Note further that, as before, $(n,w)$ and $(n,w')$ have the same balanced $4$-cycles.

\begin{lemma}
\label{le:reduce}
Suppose $(n,w)$ is an instance with $k \in \mathbb{Q}$ and $n>65k + 3$. If $\|(n,w)\|_{4-{\rm cyc}} \leq kn^3$ 
then there exists an equivalent instance $(n,w^*)$, where $w^*(e)=0$ for all but at most $32k(n+1)$ edges $e \in E(K_n)$. Moreover, we can determine $(n,w^*)$ in time $O(kn^5)$.
\end{lemma}
\begin{proof} 
From Lemma~\ref{le:S}, in time $O(kn^5)$, we can find $S \subseteq E(K_n)$ of size at most $16kn$ such all $4$-cycles in $K_n - S$ are balanced. Write $G:=K_n-S$. Let $v_0$ be a vertex of maximum degree in $G$ (so $d_G(v_0) \geq n - 1 - 32k$). Consider the weighting of $K_n$ given by
\[
w' = w - \sum_{u \in N_G(v_0)}w(uv_0)I_u.
\] 
Since $w'$ takes the form of (\ref{eq:1}), we see $(n,w)$ is equivalent to $(n,w')$. 

We show that $w'$ assigns the same weight, $\alpha$ say, to all but at most $32kn + (n-1)$ edges $e \in E(K_n)$. 
It is clear that $w'(uv_0)=0$ for every $u$ adjacent to $v_0$ in $G$. Now for any vertex $a \in V(G) \setminus \{v_0\}=V(K_n) \setminus \{v_0\}$, all edges incident to $a$ in $G$ (except possibly $av_0$) have the same weight: indeed, if $ax, ay \in E(G)$ with $x,y \not=v_0$, then $axv_0ya$ forms a balanced $4$-cycle and so we must have that $w'(ax) = w'(v_0y)+w'(ax) = w'(v_0x) + w'(ay) = w'(ay)$. This implies that $w'$ assigns the same weight to every connected component of $G- \{v_0\}$: given two edges in the same component, there is a path from one edge to the other and each pair of incident edges have the same weight, so all edges on the path have the same weight. 

Since $G - \{v_0\}$ is an $(n-1)$-vertex graph with at least $\binom{n-1}{2} - 16kn$ edges (and where $\binom{n-1}{2} \geq  16kn$ by our choice of $n$), Proposition~\ref{pr:component} implies there is a component $A$ of $G- \{v_0\}$ with at least
$\binom{n-1}{2} - 32kn = \binom{n}{2} - 32kn - (n-1)$
edges.
Thus $w'$ assigns the same weight, $\alpha$ say, to at least $e(G[A]) \geq \binom{n}{2} - 32kn - (n-1)$ edges of $K_n$. 

We set $w'':=w' - \alpha$ (i.e.\ we reduce $w'$ by $\alpha$ for all edges of $K_n$). Thus $(n,w'')$ is equivalent to $(n, w')$ (and hence to $(n,w)$) where $w''(e)=0$ for all but at most $32kn + (n-1)$ edges $e \in E(K_n)$ and $w''(e)= -\alpha$ for all edges $e \in E(G)$ incident to $v_0$. Setting $w^* = w + \alpha I_{v_0}$, we have that $(n,w^*)$ is equivalent to $(n,w'')$ (and hence to $(n,w)$) with $w^*(e)=0$ for all but at most $32kn +(n-1) - d_G(v_0) \leq 32k(n+1)$ edges $e \in E(K_n)$. So $(n,w^*)$ satisfies the requirements of the lemma.     


Finding $w^*$ takes time $O(kn^5)$; the time is dominated by the time to find $S$.
\qed
\end{proof}

We are now ready to prove Lemma~\ref{le:norm}.

\begin{proof} (of Lemma~\ref{le:norm})
By Lemma~\ref{le:reduce}, we can find in time $O(kn^5)$ an instance $(n,w')$ equivalent to $(n,w)$ such that $w'(e)=0$ for all but at most $32k(n+1) \leq 50kn$ edges of $K_n$ (for our choice of large $n$). 

Let $F$ be the spanning subgraph of $K_n$ whose edge set consists of the edges $e \in E(K_n)$ for which $w'(e) \not= 0$. Define $X = \{v \in V(K_n): d_F(v) \geq n/4 \}$.
Then $|X|n/4 \leq 2e(F) \leq 100kn$, and so $|X| \leq 400k$. We write $\overline{X}:= V(K_n) \setminus X$. 


For each $x \in X$ define $\alpha_x \in \mathbb{Q}$ by
\[
\alpha_x := \frac{1}{|\overline{X}|} \sum_{v \in \overline{X}}w'(xv)
 \hspace{1 cm} \text{and} \hspace{1 cm}
w^*:=w' - \sum_{x \in X} \lfloor \alpha_x \rfloor I_x.
\]
Thus $(n,w^*)$ is equivalent to $(n,w')$ (since it takes the form of (\ref{eq:1})) and hence to $(n,w)$ and it is easy to check that we can compute $(n,w^*)$ in $O(n^2)$ time. It remains for us to show that $w^*[K_n] \leq 4000kn$. We shall show that $w^*[K_n[\overline{X}]] \leq 2kn$ and $w^*[K_n[X,\overline{X}]] \leq 416kn$ and $w^*[K_n[X]] \leq 3208kn$, proving the lemma.

Note that $w^*(e) = w'(e)$ for all $e \in \overline{X}^{(2)}$.
Define $F^*$ to be the spanning subgraph of $K_n$ consisting of edges $e$ for which $w^*(e) \not= 0$. 
For the rest of the proof, balance is with respect to $(n,w^*)$, i.e.\ $\bal(\cdot) = \bal_{(n,w^*)}(\cdot)$.

\bigskip
\noindent
{\bf Claim~1} We have $w^*[K_n[\overline{X}]] \leq 2kn$.
\begin{proof} (of Claim~1)
Consider $e = xy \in E(F) \cap \overline{X}^{(2)} = E(F^*) \cap 
\overline{X}^{(2)}$ and define $A_{xy}$ to be the set of all $4$-cycles $uxyvu$ satisfying $u, v \in \overline{X}$ and $ux, vy, uv \not\in E(F^*) \cap X^{(2)}$, so $w^*(ux) = w^*(vy) = w^*(uv) = 0$. Note that for each $C \in A_{xy}$, $\bal(C) = |w^*(xy)|$.

There are at least $|\overline{X}| - n/4$ choices of $u \in \overline{X}$ such that $ux \not \in E(F^*)$ and  similarly $|\overline{X}| - n/4 - 1$ choices of $v \in \overline{X}$ such that $vy \not \in E(F^*)$ and $v \not = u$. Amongst these choices of $u,v$, the number of possible choices where $uv \in E(F^*) \cap \overline{X}^{(2)}$ is at most $|E(F) \cap \overline{X}^{(2)}| \leq e(F) \leq 50kn$. Hence
\begin{align*}
|A_{xy}| \geq (|\overline{X}| - n/4 - 1)^2 - 50kn 
\geq (3n/4 - 400k - 1)^2 - 50kn \geq n^2/2 
\end{align*}
where the last inequality follows by our choice of sufficiently large $n$.
Now we have
\begin{align*}
kn^3 \geq \| (n,w^*) \|_{4-{\rm cyc}} 
\geq \bal \left( \bigcup_{xy \in E(F^*) \cap X^{(2)}}A_{xy} \right) 
&\geq \sum_{xy \in E(F^*) \cap \overline{X}^{(2)}}|w^*(xy)||A_{xy}|  \\
&\geq (n^2/2)w^*[K_n[\overline{X}]],
\end{align*}
and so $w^*[K_n[\overline{X}]] \leq 2kn$.
\qed 
\end{proof}

\medskip
\noindent
{\bf Claim~2} We have $w^*[K_n[X, \overline{X}]] \leq 416kn$.
\begin{proof} (of Claim~2)
For each $x \in X$ and $ab \in  \overline{X}^{(2)}$ with $a \not= b$, let $A_{x,ab}$ be the set of $4$-cycles $C=axbca$ with $c \in \overline{X}$ satisfying $w^*(ac) = w^*(bc) = 0$ (i.e.\ $ac, bc \not \in E(F^*) \cap X^{(2)} = E(F) \cap X^{(2)}$). The number of choices for such a vertex $c$ is thus at most 
\begin{align*}
| \overline{X} | -2 - d_F(a) - d_F(b) \geq n - 2 - 400k - n/4 - n/4 \geq n/4
\end{align*}
where the last inequality follows by our choice of sufficiently large $n$. Thus we have $|A_{x,ab}| \geq n/4$ for all $x \in X$ and $ab \in \overline{X}^{(2)}$ with $a \not = b$.
Note also that, for all $C \in A_{x,ab}$, we have  $\bal(C) = |w^*(ax) - w^*(bx)|$.

Fix $x \in X$ and let $\beta_x := |\overline{X}|^{-1}\sum_{v \in \overline{X}}w^*(vx)$, noting that $\beta_x = \alpha_x - \lfloor \alpha_x \rfloor \in [0,1]$. Write $S^+ := \{v \in \overline{X}: w^*(vx) \geq \beta_x \}$ and $S^- := \{v \in \overline{X}: w^*(vx) \leq \beta_x \}$. Note for later that, since $\beta_x$ is the average of $(w^*(vx))_{v \in \overline{X}}$, we have 
\begin{equation}
\label{eq:ave}
\sum_{v \in S^+} |w^*(vx) - \beta_x| = 
\sum_{v \in S^-} |w^*(vx) - \beta_x| = 
\frac{1}{2}\sum_{v \in \overline{X}} |w^*(vx) - \beta_x|.
\end{equation}
We have
\begin{align*}
\bal \left( \bigcup_{ ab \in \overline{X}^{(2)} } A_{x, ab} \right) 
&\geq  \sum_{ a,b \in \overline{X}  } |A_{x, ab}||w^*(ax) - w^*(bx)| \\
&\geq (n/4) \sum_{ \substack{a \in S^+ \\ b \in S^-} } |w^*(ax) - w^*(bx)| \\
&= (n/4) \sum_{ \substack{a \in S^+ \\ b \in S^-} } |w^*(ax) - \beta_x| + |w^*(bx) - \beta_x|\\
&\stackrel{(\ref{eq:ave})}{=} (n/4)\left( \frac{1}{2}(|S^-| + |S^+|)
\sum_{v \in \overline{X}} |w^*(vx) - \beta_x|  \right).
\end{align*}
Using that $|S^+| + |S^-| \geq |\overline{X}|$ and $\beta_x \in [0,1]$, the last expression is bounded below by
\[
 (n |\overline{X}|/8) \sum_{v \in \overline{X}} (|w^*(vx)| - 1) 
\geq (n^2/16) \sum_{v \in \overline{X}} (|w^*(vx)| - 1)
\]
for $n$ sufficiently large. 
Finally, we have
\begin{align*}
kn^3 \geq \|(n,w)\|_{4-{\rm cyc}} 
&\geq \bal \bigg( \bigcup_{ \substack{ x\in X \\ ab \in \overline{X}^{(2)} }}A_{x, ab} \bigg)
\geq (n^2/16) \sum_{\substack{x \in X \\ v \in \overline{X}}} (|w^*(vx)| - 1) \\
&= (n^2/16) (w^*[K_n[X,\overline{X}]] - |X||\overline{X}|) \\
&\geq (n^2/16)(w^*[K_n[X,\overline{X}]] - 400kn),
\end{align*}
from which we obtain that $w^*[K_n[X,\overline{X}]] \leq 416kn$.
\qed
\end{proof}

\medskip
\noindent
{\bf Claim~3} We have $w^*[K_n[X]] \leq 3208kn$.
\begin{proof} (of Claim~3)
For each $xy \in E(F^*) \cap X^{(2)}$, define $A_{xy}$ to be the set of all $4$-cycles of the form $C=xyuvx$, where $u,v \in \overline{X}$. For fixed $xy \in E(F^*) \cap X^{(2)}$, we have
\begin{align*}
\bal(A_{xy}) 
&= \sum_{\substack{u,v \in \overline{X} \\ u \not= v}}|w^*(xy) + w^*(uv) - w^*(xv) - w^*(yu)| \\
&\geq \bigg| \sum_{\substack{u,v \in \overline{X} \\ u \not= v}}w^*(xy) + w^*(uv) - w^*(xv) - w^*(yu) \bigg| \\
&= \bigg| |\overline{X}|(|\overline{X}|-1)w^*(xy) \;+\; 2w^*(K_n[\overline{X}])
\;-\; (|\overline{X}| -1) (\beta_x + \beta_y)|\overline{X}|
 \bigg| \\
&\geq  |\overline{X}|(|\overline{X}|-1)|w^*(xy)| \;-\; 2w^*[K_n[\overline{X}]]
\;-\; 2|\overline{X}|(|\overline{X}| -1),
\end{align*}
where $\beta_x, \beta_y \in [0,1]$ are as defined in the previous claim. For our choice of large $n$, 
we have $n/2 \leq |\overline{X}| \leq n$ and so $n^2/8 \leq |\overline{X}|(|\overline{X}| - 1) \leq n^2$. By Claim~1, we have $w^*[K_n[\overline{X}] \leq 2kn \leq n^2$. Putting this together, the final expression above is at most $\frac{1}{8}n^2|w^*(xy)| - 4n^2$. Finally
\begin{align*}
kn^3 \geq \|(n,w^*)\|_{4-{\rm cyc}} 
&\geq \bal (\bigcup_{xy \in X^{(2)}}A_{xy}) 
\geq  \sum_{xy \in X^{(2)}}\frac{1}{8}n^2|w^*(xy)| - 4n^2 \\
&\geq (n^2/8) w^*[K_n[X]] - 4(400k)^2n^2 \\
&\geq (n^2/8) w^*[K_n[X]] - 400kn^3,
\end{align*}
where the last inequality holds by the choice of large $n$. Rearranging shows that $w^*[K_n[X]] \leq 3208kn$.
\qed
\end{proof}
This completes the proof of the lemma.
\qed
\end{proof}

We can now combine Lemma~\ref{le:ave} and Lemma~\ref{le:norm}  to prove the main result of this section, Theorem~\ref{th:main2}.

\begin{proof}
(of Theorem~\ref{th:main2})
Given an instance $(n,w)$ and $k \in \mathbb{N}$ with $n> 5000(k+1)$, if $\|(n,w)\|_{4-{\rm cyc}} \geq kn^3$ then by Lemma~\ref{le:ave}, in time $O(n^7)$, we can find a Hamilton cycle $H$ of $K_n$ satisfying $w(H) < dn - k$. If $\|(n,w)\|_{4-{\rm cyc}} \leq kn^3$ then in time $O(kn^5) = O(n^7)$, we can find an equivalent instance $(n,w')$ satisfying $w'[K_n] \leq 4000kn$.
\qed
\end{proof}




\section{Algorithms}
\label{se:algorithms}

In this section we prove Theorem~\ref{th:main}.
Let $(n,w,k)$ be an instance of ${\rm TSP}_{\rm{BA}}$ and let $X \subseteq V(K_n)$. An $(X)$-partial Hamilton cycle of $K_n$ is any spanning subgraph $G$ of $K_n$ that can be obtained by taking a Hamilton cycle of $K_n$ and deleting all its edges that lie in $E(K_n[\overline{X}])$. 

Suppose $G$ is a spanning subgraph of $K_n$. It is not hard to see that $G$ is an $(X)$-partial Hamilton cycle of $K_n$ if and only if
\begin{itemize}
\item[(a)] $G$ consists of vertex disjoint paths $P_1, \ldots, P_r$ for some $r$ (where we allow trivial paths);
\item[(b)] every $x \in X$ is an internal vertex of some $P_i$;
\item[(c)] no edge of $E(K_n[\overline{X}])$ is present in $G$.
\end{itemize}
For $X \subseteq V(K_n)$, we define $\overline{X}:= V(K_n) \setminus X$. 

\begin{lemma}
Let $(n,w,k)$ be an instance of ${\rm TSP}_{\rm{BA}}$ and let $X \subseteq V(K_n)$ with $t:=|X|$. We can find an $(X)$-partial Hamilton cycle of $K_n$ of minimum weight in time $O(t^3)! + O(t^3n)$.
\end{lemma}
\begin{proof}
For each $a \in X$, let $y_1^a, y_2^a, \ldots$ be an ordering of the vertices of $\overline{X}$ such that $w(ay_1^a) \leq w(ay_2^a) \leq \ldots$. For each $ab \in X^{(2)}$, let $y_1^{ab}, y_2^{ab}, \ldots$ be an ordering of the vertices of $\overline{X}$ such that $w(ay_1^{ab}) + w(by_1^{ab}) \leq w(ay_2^{ab}) + w(by_2^{ab}) \leq \ldots$.

For each positive integer $\ell$, define
\[
Y^1_{\ell}(X):=\bigcup_{a \in X}\bigcup_{i=1}^{\ell} \{y_i^a\}\subseteq \overline{X} 
\:\:\:\text{and}\:\:\:
Y^2_{\ell}(X):=\bigcup_{ab \in X^{(2)}}\bigcup_{i=1}^{\ell}
\{y_i^{ab}\}\subseteq \overline{X}
\] 
and define
\[
M^1_{\ell}(X):=\bigcup_{a \in X}\bigcup_{i=1}^{\ell} \{ay_i^a\} 
\:\:\:\text{and}\:\:\:
M^2_{\ell}(X):=\bigcup_{ab \in X^{(2)}}\bigcup_{i=1}^{\ell} \{ay_i^{ab}, by_i^{ab}\}.
\]
Finally, let $Y:=Y^1_{2t+1}(X) \cup Y^2_{2t+1}(X)$ and let $M:=M^1_{2t+1}(X) \cup M^2_{2t+1}(X)$.

\medskip
\noindent
{\bf Claim 1} Suppose $G$ is an $(X)$-partial Hamilton cycle of $K_n$ of minimum weight. Subject to this, assume further that $|E_G(X, \overline{X}) \setminus M|$ is minimised. Then $E_G(X, \overline{X}) \subseteq M$. Note that in particular, this implies that the non-trivial paths of $G$ all lie in $X \cup Y$.   
\begin{proof}(of Claim 1)
Suppose $rs \in E_G(X, \overline{X})$ with $r \in X$ and $s \in \overline{X}$ but $rs \not \in M$. 
Since $G$ is an $(X)$-partial Hamilton cycle, we have $|N_G(s)| = 1,2$ and $N_G(s) \subseteq X$ (since $G$ has no edges in $E(K_n[\overline{X}])$. Let $N_G(s)=\{r, r' \}$, where $r=r'$ if $|N_G(s)| =1$.

Let $A$  be the set of all $x \in \overline{X}$ satisfying $N_G(x)>0$;
then 
$|A| \leq 2|X| = 2t$  since $N_G(x) \subseteq X$ for all $x \in \overline{X}$ and $d_G(x) = 2$ for all $x \in X$.
 Note also that, for every $a \in \overline{X} \setminus A$, the graph $G'$ obtained by deleting the edges $rs, r's$ and replacing them with the edges $ra, r'a$ is an $(X)$-partial Hamilton cycle.

Suppose $|N_G(s)|=1$. Pick any $i \in [2t+1]$ such that $y_i^r \not \in A$ (this is possible since $|A| \leq 2t$) and replace $rs$ with $ry_i^r$ in $G$ to form $G'$. Then $G'$ is an $(X)$-partial Hamilton cycle, and $w(G') \leq w(G)$ since $w(ry_i^r) \leq w(rs)$ (by the definition of $y_i^r$). However $|E_{G'}(X, \overline{X}) \setminus M| < |E_G(X, \overline{X}) \setminus M|$ (since $rs \not \in M$ and $ry_i^r \in M$), a contradiction.

Suppose $|N_G(s)|=2$. Pick any $i \in [2t+1]$ such that $y_i^{rr'} \not \in A$ (this is possible since $|A| \leq 2t$) and replace $rs, r's$ with $ry_i^{rr'}, r'y_i^{rr'}$ in $G$ to form $G'$. Then $G'$ is an $(X)$-partial Hamilton cycle, and $w(G') \leq w(G)$ since $w(ry_i^{rr'}) + w(r'y_i^{rr'}) \leq w(rs) + w(r's)$ (by the definition of $y_i^{rr'}$. However $|E_{G'}(X, \overline{X}) \setminus M| < |E_G(X, \overline{X}) \setminus M|$ (since $rs, r's \in M$ and $ry_i^{rr'}, r'y_i^{rr'} \not\in M$), a contradiction.
\qed
\end{proof}

By Claim 1, in order to find an $X$-partial Hamilton cycle of $K_n$ of minimum weight, it is sufficient to find an $(X)$-partial Hamilton cycle of $K_n[X \cup Y]$ of minimum weight. We can do this in time $O(t^3)! + O(t^3n)$ by brute force as follows. 

We first determine $Y$; this can be done in time $O(t^3n)$ since $Y$ has size at most $O(t^3)$. We find all Hamilton cycles of $K_n[X \cup Y]$, which takes time $O(|X|+|Y|)! = O(t^3)!$.  For each such Hamilton cycle, we delete its edges in $Y^{(2)}$ to obtain an $(X)$-partial Hamilton cycle $H'$ and we determine $w(H')$ (this takes time $O(t^3)O(t^3)!=O(t^3)!$). We find the $(X)$-partial Hamilton cycle of $K_n[X \cup Y]$ of minimum weight, and its non-trivial paths form the non-trivial paths of an $(X)$-partial Hamilton cycle of $K_n$ of minimum weight.
\qed
\end{proof}

Given an instance $(n,w,k)$ of ${\rm TSP}_{\rm{BA}}$, we define the subgraph $K_w^+$ (resp.\ $K_w^-$, $K_w^0$) of $K_n$ to have vertex set $V(K_n)$ and to have edge set consisting of the  edges assigned a positive (resp.\ negative, zero) weight by $w$.

\begin{lemma}
\label{le:ker}
Suppose that $(n,w,k)$ is an instance of ${\rm TSP}_{\rm{BA}}$ and $X \subseteq V(K_n)$ with $|X|=t$ satisfies the following properties:
\begin{itemize}
\item[(a)] $w(e) \geq 0$ for all $e \in E(K_n[\overline{X}])$;
\item[(b)] $\delta(K_w^0[\overline{X}]) \geq \frac{1}{2}|\overline{X}| + 4t$.
\end{itemize}
Then we can find a minimum weight Hamilton cycle $H^*$ of $K_n$ 
in time $O(t^3)! + O(t^3n + n^3)$.
\end{lemma}
\begin{proof}
In time $O(t^3)! + O(t^3n)$, we can find an $(X)$-partial Hamilton cycle $G^*$ of $K_n$ of minimum weight. Let $P_1, \ldots, P_r$ be the non-trivial paths of $G^*$ and note that $r \leq t$. Let $a_i,b_i \in \overline{X}$ be the end-points of $P_i$ and let $e_i=a_ib_i$. 

Let $X'$ be the set of internal vertices of $P_1, \ldots, P_r$. Thus $X \subseteq X'$ and it is not hard to see that $|X'| \leq 2|X| = 2t$. Thus 
\begin{align*}
\delta(K_w^0[\overline{X'}]) \geq \delta(K_w^0[\overline{X}]) - |\overline{X} \setminus \overline{X'}| 
&\geq \delta(K_w^0[\overline{X}]) - 2t 
\geq \frac{1}{2}|\overline{X}| + 4t - 2t \\
&\geq \frac{1}{2}|\overline{X'}| + 2t
\geq \frac{1}{2}|\overline{X'}| + \frac{3}{2}r.
\end{align*}
By an algorithmic version of Dirac's Theorem (see e.g.\ \cite{KuhOstPat} Lemma 5.11), we can find, in time $O(n^3)$, a Hamilton cycle $H$ of $K_w^0[\overline{X'}] \cup \{e_1, \ldots, e_r \}$ that contains all the edges $e_1, \ldots, e_r$. Replacing each edge $e_i$ with the path $P_i$ gives a Hamilton cycle $H^*$ of $K_n$, which we claim has minimum weight.

Indeed, let $H$ be any Hamilton cycle of $K_n$. 
Let $G$ be obtained from $H$ by deleting all the edges of $H$ in $\overline{X}^{(2)}$; thus $G$ is an $(X)$-partial Hamilton cycle, so $w(G^*) \leq w(G)$. Note also by condition (i) that $w(H[\overline{X}]) \geq 0$, whereas by construction $w(H^*[\overline{X}])=0$. Therefore
\begin{align*}
w(H) 
= w(G) + w(H[\overline{X}]) 
\geq w(G^*) + w(H^*[\overline{X}])
=w(H^*)
\end{align*}
showing that $H^*$ is a minimum weight Hamilton cycle of $K_n$.
\qed
\end{proof}

Finally we can prove Theorem~\ref{th:main}.

\begin{proof} (of Theorem~\ref{th:main})
We assume $n > ck$ for a sufficiently large constant $c$; otherwise we can find a Hamilton cycle of minimum weight by brute force. We describe the steps of the algorithm and give the running time in brackets.
\begin{itemize}
\item[1.] Given $(n,w,k)$, by Theorem~\ref{th:main2}, either we can output a Hamilton cycle $H^*$ satisfying $w(H^*) < dn -k$ (and we are done) or we we can find an equivalent instance $(n,w',k)$ where $w'[K_n] \leq 4000kn$.
(time $O(n^7)$)
\item[2.] Find a maximum matching $Q$ of $K_{w'}^-$, the subgraph of $K_n$ containing the edges assigned a negative weight $w'$. Let $q$ be the number of edges of $Q$ and set $d':=w'(K_n)/ \binom{n}{2}$. (time $O(n^3)$; see \cite{Edm} for a polynomial-time maximum matching algorithm)
\item[3.] If $q > 10^5k$ then $\mathbb{E}(w'(\tilde{H}^Q)) \leq d'n-k$ (see Claim 1) and so by Lemma~\ref{le:derand2}, we can find a Hamilton cycle $H^*$ of $K_n$ in time $O(n^5)$ satisfying $w'(H^*) \leq \mathbb{E}(w'(\tilde{H}^Q)) \leq d'n - k$ and hence $w(H^*) \leq dn -k$. We output $H^*$ and stop. (time $O(n^5)$)
\item[4.] If $q \leq 10^5k$, then construct the set $X \subseteq V(K_n)$ such that $X = V(Q) \cup \{v \in V(K_n) \mid d_{K_{w'}^+}(v) \geq \frac{1}{4}n \}$. 
(time $O(n^2)$)
\item[5.] Using the properties of $X$ proved in Claim 2, we can apply Lemma~\ref{le:ker} to $(n,w',k), X$ to find a Hamilton cycle of minimum weight in $K_n$ (w.r.t. $w'$). If $w'(H^*) \leq d'n-k$ then output $H^*$ and note $w(H^*) \leq dn - k$. Otherwise we conclude there is no Hamilton cycle beating the average by at least $k$.  (time $O(k^3)! + O(k^3n + n^3)$) 
\end{itemize}

\medskip
\noindent
{\bf Claim 1} If $q > 10^5k$ then $\mathbb{E}(w'(\tilde{H}^Q)) \leq d'n-k$.
\begin{proof} (of Claim 1)
If $q > 10^5k$ we have
\begin{align*}
\mathbb{E}(w'(\tilde{H}^Q)) 
&= w'(Q) + \sum_{e \in J(Q)} \mathbb{P}(e \in \tilde{H}^Q)w'(e)\\ 
&\leq -10^5k + (4/(n-2)) \sum_{e \in J(Q)} w'(e) \\
&\leq -10^5k + (4/(n-2))4000kn \\
&\leq -50000k, 
\end{align*}
where we have used that $w'(Q)\leq -q$ (since $Q \subseteq K_{w'}^-$) and $\mathbb{P}(e \in \tilde{H}^Q) \leq 4/(n-2)$ (using (\ref{eq:prob})) for the first inequality and where we have used that $w'[K_n] \leq 4000kn$ for the second inequality, and that $n$ is large for the third inequality.
Note also that
\begin{align*}
d'n - k = \binom{n}{2}^{-1}w'(K_n)n - k 
\geq -4000k\frac{n}{n-1} - k 
&\geq - 50000k \\
&\geq \mathbb{E}(w'(\tilde{H}^Q)),
\end{align*}
where the we have used that $n$ is large enough for the second inequality. This gives the desired result.
\qed
\end{proof}

\medskip
\noindent
{\bf Claim 2} 
For the set $X$ defined in Step 4, we have
\begin{itemize}
\item[(a)] $t:=|X| \leq 3\cdot 10^5k$;
\item[(b)] $w'(e) \geq 0$ for all $e \in E(K_n[\overline{X}])$;
\item[(c)] $\delta(K_{w'}^0[\overline{X}]) \geq \frac{1}{2}|\overline{X}| + 4t$.
\end{itemize}
\begin{proof} (of Claim 2)
(a) We have $X = V(Q) \cup S$ where $S:= \{v \in V(K_n) \mid d_{K_{w'}^+}(v) \geq \frac{1}{4}n \}$. We show that $|S| \leq 10^5k$ and since $|V(Q)| \leq 2q \leq 2\cdot 10^5k$, we have $|X| \leq 3\cdot 10^5k$ as required. 

Observe that
\begin{align*}
\frac{1}{4}n|S| \leq 2e(K_{w'}^+) \leq 8000kn, 
\end{align*}
which implies $|S| \leq 32000k \leq 10^5k$, as required.

(b) follows from the construction of $X$, since $X$ contains all vertices of a maximum matching from the graph of negatively weighted edges. 

To prove (c) observe that for each $x \in \overline{X}$
\begin{align*}
d_{K_{w'}^0[\overline{X}]}(x) &= (n-1) - |X|
- d_{K_{w'}^+[\overline{X}]}(x)
- d_{K_{w'}^-[\overline{X}]}(x) \\
&\geq (n-1) -3 \cdot 10^5 k - (n/4) - 0 \\
&\geq (3/5)n - 3\cdot 10^5k  \\
&\geq (1/2)n + 4t,
\end{align*}
where the last inequality holds for $n$ large enough.
\qed
\end{proof}
This completes the proof of the theorem.
\qed
\end{proof}

\section{Concluding Remarks}

We believe that an analogue of our result ought to hold for the Asymmetric Travelling Salesman Problem. Formally, let $\overset\leftrightarrow{K}_n$ denote the complete directed graph on $n$ vertices, so that there are two edges, one in each direction, between each pair of vertices.

\bigskip
\noindent
\textsc{Asymmetric Travelling Salesman Problem Below Average (${\rm ATSP}_{\rm{BA}}$)} 

\medskip
\noindent
\begin{tabular}{p{1.7cm}p{11cm}}
\textit{Instance}\,:& $(n,w,k)$, where $n,k \in \mathbb{N}$ and $w: E(\overset\leftrightarrow{K}_n) \rightarrow \mathbb{Z}$\\
\textit{Question}\,:&Is there a directed Hamilton cycle $H^*$ of $\overset\leftrightarrow{K}_n$ satisfying $w(H^*) \leq dn - k$, (where $d$ is the average weight of an edge of $\overset\leftrightarrow{K}_n$)? 
\end{tabular}

\bigskip

\begin{problem}
Is ${\rm ATSP}_{\rm{BA}}$ fixed parameter tractable when parameterised by $k$?
\end{problem}
We believe the answer to this question is yes. However, several of the methods in this paper do not generalise in a straightforward way to directed graphs, so we expect that several new ideas will be needed to solve the problem above.

\end{document}